%% file: emd-april7.tex
\documentclass[11pt]{article}
\usepackage{verbatim}

\usepackage{amsmath}

\usepackage[noend]{algorithm2e}
\incmargin{1em}
\restylealgo{unboxed}
\linesnumbered
\dontprintsemicolon
\SetArgSty{rm}
\SetTitleSty{textsc}{11pt}
\SetKwComment{Comment}{ // }{}

\usepackage{amssymb}

\newcommand{\ra}{\rightarrow}

\newcommand{\A}{\mathcal{A}}

\newcommand{\M}{\mathcal{M}}

\newcommand{\eps}{\varepsilon}

\newcommand{\diam}{\textrm{\rm diam}}

\newcommand{\tO}{\tilde{O}}
\renewcommand{\bar}[1]{\overline{#1}}

\newcommand{\mnote}[1]{}

\title{Sublinear Time Algorithms for Earth Mover's Distance}
\author{Khanh Do Ba\\MIT, CSAIL\\doba@mit.edu \and 
Huy L. Nguyen\\MIT\\hlnguyen@mit.edu \and
Huy N. Nguyen\\MIT, CSAIL\\huy2n@mit.edu \and
Ronitt Rubinfeld\\MIT, CSAIL\\ronitt@csail.mit.edu}

\begin{document}
\input{preamble.tex}
\maketitle

\begin{abstract}
	We study the problem of estimating the Earth Mover's Distance (EMD) between
	probability distributions when given access only to samples. We give closeness
	testers and additive-error estimators over domains in $[0,\Delta]^d$, with
	sample complexities independent of domain size -- permitting the testability
	even of continuous distributions over infinite domains. Instead, our algorithms depend on
	other parameters, such as the diameter of the domain space, which may be
	significantly smaller. We also prove lower bounds showing the dependencies
	on these parameters to be essentially optimal. Additionally, we consider whether
	natural classes of distributions exist for which there are
	algorithms with better dependence on the dimension, and show that for
	highly clusterable data, this is indeed the case. Lastly, we consider a variant
	of the EMD, defined over tree metrics instead of the usual $\ell_1$ metric, and
	give optimal algorithms.
\end{abstract}
\section{Introduction}

In traditional algorithmic settings, algorithms requiring linear time and/or space
are generally considered to be highly efficient; sometimes even
polynomial time and space requirements are acceptable. However, today this is often no
longer the case. With data being generated at rates of terabytes a second,
sublinear time algorithms have become crucial in many applications. In the increasingly important
area of massive data algorithmics, a number of models have been proposed and
studied to address this. One of these arises when the data can be naturally viewed
as a probability distribution (e.g., over IP addresses, or items sold by an online
retailer, etc.) that allows i.i.d. samples to be drawn from it. This is the model
on which this paper focuses.

Perhaps the most fundamental problem in this model is that of testing whether two
distributions are close. For instance, if an online retailer such as Amazon.com
wishes to detect changes in consumer habits, one way of doing so might be to see if
the distribution of sales over all offered items this week, say, is significantly
different from last week's distribution. This problem has been studied extensively, mostly
under the $\ell_1$ and $\ell_2$ distances, and algorithms sublinear in time and
sample complexity exist to distinguish
whether two distributions are identical or $\eps$-far from each other~\cite{L1tester,GR,valiant}.
However, under the $\ell_1$ distance, for instance, the sample complexity, though sublinear, can be no
smaller than $n^{2/3}$ (where $n$ is the domain size), which may be
prohibitively large~\cite{L1tester,valiant}.

Fortunately, in many situations there is a natural metric on the underlying domain,
under which nearby points should be treated as ``less different'' than faraway
points. This motivates a metric known as Earth Mover's Distance (EMD), first
introduced in the vision community as a measure of (dis)similarity between images
that more accurately reflects human perception than more traditional $\ell_1$
~\cite{peleg}. It has since proven to be important in computer graphics and vision
~\cite{rubner2,rubner1,rubner3,cohen,ruzon1,ruzon2,rubner4}, and
has natural applications to other areas of computer science. As a result, its
computational aspects have recently drawn attention from the algorithms community
as well~\cite{alex,indyk,charikar,indyk2}. However, previous work has generally focused
on the more classical model of approximating the EMD when given the input distributions
explicitly; that is, when the exact probability of any domain element can be queried.
As far as we know, no work has been done on estimating and closeness
testing of EMD when given access only to i.i.d. samples of the distributions.

In this model, it is easy to see that we cannot hope to compute a multiplicative approximation,
even with arbitrarily many samples, since that would require us to distinguish between
arbitrarily close distributions and identical ones. However, if we settle for additive error, we
show in this paper that, in contrast to the $\ell_1$ distance, we can estimate EMD using a number
of samples {\em independent}
of the domain size. Instead, our sample complexities depend only on the {\em
diameter} of the domain space, which in many cases can be significantly smaller.
The consequence is that this allows us to effectively deal with distributions over
extremely large domains, and even, under a natural generalization of EMD,
continuous distributions over infinite domains.

Specifically, if $p$ and $q$ are distributions over $M \subset [0,\Delta]^d$ (where $d$ is a
constant), we can
\begin{itemize}
	\item estimate $EMD(p,q)$ to within an additive error of $\eps$ with
		$\tilde{O}((\Delta/\eps)^{d + O(1)})$ samples,
	\item distinguish whether $p=q$ or $EMD(p,q) > \eps$ with
		$\tilde{O}((\Delta/\eps)^{2d/3 + O(1)})$ samples, and
	\item if $q$ is known, distinguish whether $p=q$ or $EMD(p,q) > \eps$ with
		$\tilde{O}((\Delta/\eps)^{d/2 + O(1)})$ samples.
\end{itemize}
We also give lower bounds that imply these results to be essentially optimal
(up to small $poly(\Delta/\eps)$ factors). In the case of $d=1$ or $2$, upper
and lower bounds for both testers become $\Theta((\Delta/\eps)^2)$.

Additionally, we consider assumptions on the data that might make the problem
easier, and give an improved algorithm in the case our input distributions is highly
clusterable. Finally, it is natural to consider the EMD over domains endowed with a
metric other than $\ell_1$ distance. We give an optimal (upto polylogarithmic factors)
algorithm for estimating EMD over tree metrics.

\section{Preliminaries}

\subsection{Earth Mover's Distance}

We start with the following definition.

\begin{definition}
	A {\em supply-demand network} is a directed bipartite graph $G=(S\cup T,E)$
	consisting of supply vertices $S$ and demand vertices $T$, with supply
	distribution $p$ on $S$ and demand distribution $q$ on $T$, and edge set
	$E=S\times T$,  with associated weights $w:E \ra \R^+$. A {\em satisfying flow} for $G$ is a
	mapping $f:E\ra\R^+$ such that for each $s \in S$ and each $t \in T$,
	\begin{eqnarray*}
		\sum_{t'\in T} f((s,t')) &=& p(s),\text{ and}\\
		\sum_{s'\in S} f((s',t)) &=& q(t).
	\end{eqnarray*}
	The cost of satisfying flow $f$ is given by
	\[
		C(f) = \sum_{e\in E} f(e)w(e).
	\]	
\end{definition}

We define the Earth Mover's Distance (EMD) as follows.

\begin{definition}\label{def:flow}
	Let $p$ and $q$ be probability distributions on a finite metric space $(M,\delta)$.
	Then let $G$ be the supply-demand network given by supply vertices
	$S=\{s_x~|~x\in M\}$ and demand vertices $T=\{t_x~|~x\in M\}$, with supply distribution
	$\hat{p}: s_x \mapsto p(x)$ and demand distribution $\hat{q}: t_x \mapsto q(x)$, and
	edge weights $w: (s_x,t_y) \mapsto \delta(x,y)$. Define $EMD(p,q)$ to be the minimum
	cost of all satisfying flows for $G$.
\end{definition}

It is straightforward to verify that the EMD as defined above is a metric on all
probability distributions on $M$. Note, moreover, that to upperbound the EMD it
suffices to exhibit any satisfying flow. 

Since the magnitude of the EMD depends on the distances in the underlying metric
space $M$, we must assume that $M$ is of bounded diameter. In particular, we will in
this paper focus mostly on the case where $M \subset [0,\Delta]^d$, for some $\Delta > 0$,
endowed with the $\ell_1$ metric.

Finally, we define a closeness tester for the EMD in the usual way as follows.

\begin{definition}
  Let $p$, $q$ be two distributions on (finite) metric space $M$. An
  {\em EMD-closeness tester} is an algorithm which takes as input samples from $p$ and $q$,
  together with a real number $\eps>0$, and guarantees that
  \begin{itemize}
    \item[(1)] if $p = q$, then it accepts with probability at least $2/3$, and
    \item[(2)] if $EMD(p,q) > \eps$, then it rejects with probability at least $2/3$.
  \end{itemize}
\end{definition}

Note that it is also possible to define EMD for any two probability measures $p$ and $q$ in an infinite (continuous) metric space $(\mathcal{M},\delta_\mathcal{M})$, $\mathcal{M} \subset [0,\Delta]^d$, by using Wasserstein metric: $$ EMD(p,q) =  \inf_{\gamma \in \Gamma(p,q)}{\int_{\mathcal{M} \times \mathcal{M}}{\delta_\mathcal{M}(x,y)d\gamma(x,y)}}  $$  where $\Gamma(p,q)$ denotes the collection of all measures on $\mathcal{M} \times \mathcal{M}$ with marginals $p$ and $q$ on the first and second factors respectively. By using $\eps/4$-net, $(\mathcal{M},\delta_\mathcal{M})$ can be discretized into a finite metric $(M,\delta)$ in such a way that the $EMD$ distance of any two probability measures only increases or decreases at most $\eps/2$. Therefore, an $EMD$-closeness tester with additive error $\eps/2$ for $(M,\delta)$ is also a valid $EMD$-closeness tester for $(\mathcal{M},\delta)$ with additive error $\eps$.

\subsection{Properties of EMD}

Let us get a firmer handle on the EMD by relating it to $\ell_1$ distance, via the
following lemmas.

\begin{lemma}\label{lem:L1/2}
        If $p$ and $q$ are distributions on $(M,\delta)$, there exists a minimum cost
        satisfying flow $f$ from $S$ to $T$ (as defined in Def. \ref{def:flow}) such that the total
        amount sent by $f$ across edges with non-zero cost is exactly $||p-q||_1/2$.
\end{lemma}

\begin{proof}
        Let $a = \sum_x \min\{p(x),q(x)\}$ and $A = \sum_x \max\{p(x),q(x)\}$, so that
        $A-a = ||p-q||_1$ and $A+a = 2$. Observe that the total amount sent from $S$ to
        $T$ is 1, and the maximum possible amount sent across edges with zero cost is $a$, which
        leaves at least $1-a = ||p-q||_1/2$ to be sent across non-zero cost edges.
        
        On the other hand, suppose that for some $x$, an optimal flow through the edge
        $(s_x,t_x)$ is less than $\min\{p(x),q(x)\}$.
        Then there exist points $y$ and $z$ such that $f$ sends at least $\alpha>0$ from $s_x$
        to $t_y$ and from $s_z$ to $t_x$. We can replace this partial flow, which costs
        $\alpha(\delta(x,y) + \delta(x,z))$, with one sending $\alpha$ from $s_x$ to $t_x$
        and from $s_z$ to $t_y$, which costs $\alpha(\delta(y,z))$. Doing so, we will not
        increase the total cost (by triangle inequality), nor will we affect the rest of the
        flow. We can therefore repeated the procedure until we obtain an optimal flow that
        saturates every zero-edge.
\end{proof}

\begin{corollary}\label{cor:EMDvsL1}
        If $p$ and $q$ are distributions on $M$, where $M$ has minimum
        distance $\delta$ and diameter $\Delta$, then
        \[
                \frac{||p-q||_1}{2}\cdot\delta \le EMD(p,q) \le \frac{||p-q||_1}{2}\cdot\Delta.
        \]
\end{corollary}

\begin{lemma}\label{lem:buckets}
        Let $p$ and $q$ be distributions on $M$ with diameter $\Delta$, and
        $\M = \{M_1,\dots,M_k\}$ be a partition of $M$ wherein $\diam(M_i)\le\Gamma$ for
        every $i\in[k]$. Let $P$ and $Q$ be distributions on $\M$
        induced by $p$ and $q$, resp. Then $EMD(p,q) \le \frac{||P-Q||_1}{2}\cdot\Delta +
        \Gamma$.
\end{lemma}

\begin{proof}
        Let us define distribution $p'$ by moving some of the probability weight of $p$
        {\em between} $M_i$'s (taking from and depositing anywhere within the respective
        $M_i$'s) in such a way that $p'$ induces $Q$ on $\M$. This is effectively a flow
        from $P$ to $Q$ where all distances are bounded by $\Delta$, so by Lemma
        \ref{lem:L1/2} it can be done at cost at most $\frac{||P-Q||_1}{2}\cdot\Delta$.
        It follows that $EMD(p,p') \le \frac{||P-Q||_1}{2}\cdot\Delta$.
        
        Then, having equalized the probability weights of each $M_i$, let us move the
        probability weight of $p'$ {\em within} each $M_i$ to precisely match $q$. This
        might require moving everything (i.e., 1), but the distance anything is moved is
        at most $\Gamma$, so $EMD(p',q) \le \Gamma$ and the lemma follows by triangle
        inequality.
\end{proof}

\subsection{Some tools from previous work}

Here, for completeness, we state some results from previous work that we will make use
of later. First, we define some useful distributions described in \cite{indyk}
for testing closeness between distributions over subsets of $[0, \Delta]^d$.

\begin{definition}
Given distribution $p$ over $M \subset [0, \Delta]^d$ and a positive integer $i$, let $G^{(i)}$ be a grid with side length $\frac{\Delta}{2^i}$ over $[0, \Delta]^d$ centered at the origin. Define the {\em $i$-coarsening of $p$}, denoted $p^{(i)}$, to be the distribution over the grid cells of $G^{(i)}$ such that, for each grid cell c of $G^{(i)}$, $p^{(i)}(c) = \sum_{u\in c} p(u)$.
\end{definition}

The $p^{(i)}$'s can be thought of as coarse approximations of $p$
where all points in the same grid cell are considered to be the same point.
We then have the following lemma from \cite{indyk} relating the EMD of two distributions
to a weighted sum of the $\ell_1$ distances of their $i$-coarsenings.
\begin{lemma}
\label{lemma:bindiv}
For any two distributions $p$ and $q$ over $M \subset [0, \Delta]^d$,
\[
  EMD(p, q) \leq d \left( \sum_{i=1}^{\log{(2\Delta d/\eps)}}
  		\frac{\Delta}{2^{i-1}}\cdot \left\| p^{(i)} - q^{(i)}\right\|_1  \right) + \frac{\eps}{2}.
\]
\end{lemma}

Having established the various relationships between the EMD and the $\ell_1$ distance,
we will make use of the result from \cite{L1tester} below for testing closeness of distributions
in $\ell_1$ as a subroutine.

\begin{theorem}\label{complexL1tester}
  Given access to samples from two distributions $p$ and $q$ over $M$, where $|M|=n$,
  there exists an algorithm that takes $O(n^{2/3}\eps^{-4}\log n\log(1/\delta))$
  samples and (1) accepts with probability at least $1-\delta$ if $p=q$
  and (2) rejects with probability at least $1-\delta$ if $\left\|p-q\right\|_1 > \eps$.
\end{theorem}

Alternatively, via a simple Chernoff bound analysis similar to \cite{L1indep},
we can show that a whole distribution can be approximated efficiently, giving
us another closeness tester for $\ell_1$.

\begin{lemma}
  Given access to samples from a distribution $p$ over $M$, where $|M|=n$, and
  parameters $\eps,\delta,t > 0$, there exists an algorithm
  that takes $O(t^{-1}\eps^{-2}\log n\log(1/\delta))$ samples and outputs a
  distribution $\tilde{p}$ over $M$ such that with probability at least $1-\delta$
  \[
  	(1-\eps)\max\{p(i),t\} \leq \tilde{p}(i) \leq (1+\eps)\max\{p(i),t\}
  \]
  for every $i \in M$.
\end{lemma}

As a result, we can simply estimate $p$ and $q$ with $\tilde{p}$ and $\tilde{q}$,
respectively, and compute their $\ell_1$ distance to get the following alternative
tester, which gives us a different trade-off between $n$ and $\eps$ that will
become useful later.

\begin{theorem}\label{naiveL1tester}\mnote{combine w Th. \ref{complexL1tester}?}
  Given access to samples from two distributions $p$ and $q$ over M, where $|M|=n$,
  there exists an algorithm that takes $O(n\eps^{-2}\log n\log(1/\delta))$
  samples and (1) accepts with probability at least $1-\delta$ if $p=q$
  and (2) rejects with probability at least $1-\delta$ if $\left\|p-q\right\|_1 > \eps$.
\end{theorem}

\section{Closeness tester}\label{section:41}

In this section, we consider the EMD-closeness testing problem when the
domain is $M \subset [0, \Delta]^d$. The main idea behind the algorithm is to
embed EMD into the $\ell_1$ metric and use an $\ell_1$-closeness tester (Theorems
\ref{complexL1tester} and \ref{naiveL1tester}) to test the
resulting distributions. Recall from the preliminaries, $p^{(i)}$ and $q^{(i)}$ are
the $i$-coarsening approximations of $p$ and $q$. We have the following algorithm,
where the subroutine \FuncSty{$\ell_1$-Closeness-Tester}$(p,q,\eps,\delta)$ is an
$\ell_1$-closeness tester on distributions $p$ and $q$ with distance parameter $\eps$
and failure probability $\delta$.

\FuncSty{EMD-Closeness-Tester$(p, q, \eps)$}\\
\begin{algorithm}[H]

  \For {$i = 1$ \KwTo $\log(2\Delta d/\eps)$} {
    \If {\FuncSty{$\ell_1$-Closeness-Tester}$(p^{(i)}$, $q^{(i)},
    				\frac{\eps 2^{i-2}}{\Delta d \log(2\Delta d/\eps)},
    				\frac{1}{3\log(2\Delta d/\eps)})$ rejects} {
      \KwSty{reject}
    }
  }
  \KwSty{accept}
\end{algorithm}

Note that our subroutine takes advantage of whichever tester (Theorem \ref{complexL1tester}
or \ref{naiveL1tester}) requires fewer samples. Specifically, when $d$ is small,
the domains of the $i$-coarsenings of $p$ and $q$ are small and $\eps$ is the bottleneck,
so we use Theorem \ref{naiveL1tester}. On the other hand, if $d$ is large, the sizes of
these domains become the bottleneck and we use Theorem \ref{complexL1tester}. This gives us
the following theorem.

\begin{theorem}
  The above is an EMD-closeness tester for distributions over $M \subset [0,\Delta]^d$ that
  takes $\tilde{O}((2\Delta d/\eps)^{2d/3})$ samples when $d\geq 6$ and
  $\tilde{O}((\Delta/\eps)^{2})$ samples when $d\le2$.
\end{theorem}
\begin{proof}
  If $p=q$, then $p^{(i)}=q^{(i)}$ for all $i$, so by the union bound, the probability that
  the algorithm rejects is at most $\log{\frac{2\Delta d}{\eps}} \cdot \frac{1}{3\log\frac{2\Delta d}{\eps}} = \frac{1}{3}$. 

  If, on the other hand, $EMD(p,q) > \eps$, then by Lemma {\ref{lemma:bindiv}},
  \[
    d\left(\sum_{i=1}^{\log(2\Delta d/\eps)} \frac{\Delta}{2^{i-1}}\cdot \left\| p^{(i)} - q^{(i)}\right\|_1 \right)
    	> \frac{\eps}{2}.
  \]
  It follows by the pigeonhole principle that there exists an index $i$ such that
  \[
  	\left\| p^{(i)} - q^{(i)}\right\|_1
  		> \frac{\eps2^{i-2}}{\Delta d\log(2\Delta d/\eps)}.
  \]
  Hence, for that index $i$, the $\ell_1$-closeness tester in Step 2 will reject (with
  probability $2/3$).
  
  Now let us analyze the number of samples the algorithm needs. In the $i^{th}$ iteration of
  the main loop, $p^{(i)}$ and $q^{(i)}$ has a domain with $n_i=2^{di}$ elements, and we
  need to run an $\ell_1$-closeness tester with a distance parameter of
  $\eps_i = \frac{\eps2^{i-2}}{\Delta\log(2\Delta d/\eps)}$. Consider the following
  two cases:
  \begin{itemize}
    \item $d \geq 6$: Using the algorithm of Theorem \ref{complexL1tester}, we get a sample
    complexity of
    \[
    	\tilde{O}(n_i^{2/3}\eps_i^{-4})
    		= \tilde{O}\left(2^{(2d/3-4)i}\left(\frac{4\Delta d\log(2\Delta d/\eps)}{\eps}\right)^4\right).
    \]
    This quantity is maximized when $i = \log(2\Delta d/\eps)$, which gives us a total
    complexity of $\tilde{O}((2\Delta d/\eps)^{2d/3})$.
    
    \item $d \leq 2$: Using the algorithm of Theorem \ref{naiveL1tester}, we get a sample
    complexity of
    \[
    	\tilde{O}(n_i\eps_i^{-2})
    		= \tilde{O}\left(2^{(d-2)i}\left(\frac{4\Delta d\log(2\Delta d/\eps)}{\eps}\right)^2\right).
    \]
    This quantity is maximized when $i=1$, giving us a total complexity of
    $\tilde{O}((\Delta/\eps)^2)$.
  \end{itemize}
\end{proof}

If one of the distributions is explicitly known, we can use the corresponding $\ell_1$-closeness
tester with sample complexity $O(n^{1/2}\eps^{-2}\log n)$ from \cite{L1indep} to
similarly get the following theorem.

\begin{theorem}
	There exists an EMD-closeness tester for distributions over $M \subset [0,\Delta]^d$,
	where one is explicitly known, that takes $\tilde{O}((2\Delta/\eps)^{d/2})$ samples
	when $d\ge4$.
\end{theorem}

\section{Additive-error estimation}

We have seen that in $\ell_1$-closeness testing, sometimes it is to our advantage to
simply estimate each probability value, rather than use the more sophisticated
algorithm of \cite{L1tester}. This seemingly naive approach has another advantage:
it gives an actual numeric estimate of the distances, instead of just an
accept/reject answer. Here, we use this approach to obtain an additive approximation
of the EMD of two unknown distributions over $M \subset [0,\Delta]^d$ as follows.

\FuncSty{EMD-Approx}$(p,q,\eps)$\\
\begin{algorithm}[H]
	Let $G$ be the grid on $[0,\Delta]^d$ with side length $\frac{\eps}{4d}$,
		and let $P$ and $Q$ be the distributions induced by $p$ and $q$ on $G$, with
		weights in each cell concentrated at the center\;
	Take $O((4d\Delta/\eps)^{d+2})$ samples from $P$ and $Q$, and let $\tilde{P}$
		and $\tilde{Q}$ be the resulting empirical distributions\;
	\Return{$EMD(\tilde{P},\tilde{Q})$}
\end{algorithm}

\begin{theorem}
	\FuncSty{EMD-Approx} takes $O((4d\Delta/\eps)^{d+2})$ samples from $p$ and
	$q$ and, with probability $2/3$, outputs an $\eps$-additive approximation of $EMD(p,q)$.
\end{theorem}

\begin{proof}
Note that a sample from $p$ or $q$ gives us a sample from $P$ or $Q$, respectively, so it remains to prove correctness.
Observe that $|G| = (4d\Delta/\eps)^d$, so $G$ has $2^{(4d\Delta/\eps)^d}$ subsets. By the Chernoff bound, with the $O((4d\Delta/\eps)^{d+2})$ samples from $p$, we can guarantee for each $S \subseteq G$ that
$|P(S) - \tilde{P}(S)| > \frac{\eps}{4d\Delta}$ with probability at most $2^{-(4d\Delta/\eps)^d}/3$. By the union bound, with probability at least $2/3$, all subsets of $G$ will be approximated to within an additive $\frac{\eps}{4d\Delta}$. In that case,
\begin{eqnarray*}
	\|P - \tilde{P}\|_1
		&=&\sum_{c\in G} |P(c) - \tilde{P}(c)|\\
		&=& 2\max_{S \subseteq G}{|P(S) - \tilde{P}(S)|} \le \frac{\eps}{2d\Delta}.
\end{eqnarray*}

We then have, by Corollary \ref{cor:EMDvsL1},
$EMD(P,\tilde{P}) \le \eps/4$. Further, since each cell has radius $\eps/4$,
we have $EMD(p,P) \le \eps/4$, giving us by the triangle inequality, $EMD(p,\tilde{P}) \le \eps/2$.
Similarly, $EMD(q,\tilde{Q}) \le \eps/2$, so again by triangle inequality, we get
\[
	|EMD(p,q) - EMD(\tilde{P},\tilde{Q})| \le EMD(p,\tilde{P}) + EMD(q,\tilde{Q}) = \eps,
\]
completing our proof.
\end{proof}

\section{Lower bounds}

We can show that our tester is optimal for the 1-dimensional domain by a simple argument:

\begin{theorem}
 	Let $\A$ be an EMD-closeness tester of distributions over any domain $M \subset [0, \Delta]$.
 	Then $\A$ requires $\Omega((\Delta/\eps)^2)$ samples.
\end{theorem}

\begin{proof}
Consider two distributions $p$ and $q$ over the domain $\{0,\Delta\}$, where $p$ is the distribution that puts the weight of $\frac{1}{2}$ at $0$ and $\Delta$ and $q$ is the distribution that puts the weight of $1/2 + \eps/\Delta$ at $0$ and the weight of $1/2 - \eps/\Delta$ at $\Delta$. Clearly $EMD(p,q) = \eps$, and it is a classic result that distinguishing $p$ from $q$ requires $\Omega((\Delta/\eps)^2)$ samples.
\end{proof}

Clearly this also implies the same lower bound for 2-dimensional domains, making our
algorithm optimal in those cases. Next we prove that our $d$-dimensional tester is
also essentially optimal in its dependence on $\Delta/\eps$.

\begin{theorem}
	There is no EMD-closeness tester that works on any $M \subset [0,\Delta]^d$ that
	takes $o((\Delta/\eps)^{2d/3})$ samples.
\end{theorem}

\begin{proof}
	Suppose $\A$ is an EMD-closeness tester that requires only $o((\Delta/\eps)^{2d/3})$
	samples. Then consider the following $\ell_1$-closeness tester for $\eps=1$:

	\FuncSty{$\ell_1$-Tester$(p,q,\eps=1)$}\\
	\begin{algorithm}[H]
		Let $G$ be a grid on $[0,\Delta]^d$ with side length $\Delta n^{-1/d}$\;
		Let $f$ be an arbitrary injection from $[n]$ into the lattice points of $G$\;
		Let $P$ and $Q$ be distributions on the lattice points of $G$ induced by
			$f$ on $p$ and $q$, resp.\;
		\Return{$\A(P,Q,\frac{1}{2} \Delta n^{-1/d})$}\;
	\end{algorithm}

	Correctness is easy to see: if $p=q$, then clearly $P=Q$ as well and the tester
	accepts; alternatively, if $\|p-q\|_1 = 1$, then by Corollary \ref{cor:EMDvsL1}
	and the observation that $\|P-Q\|_1 = \|p-q\|_1$,
	\[
		EMD(P,Q) \ge \frac{\|P-Q\|_1}{2} \cdot \Delta n^{-1/d}
			=      \frac{1}{2} \Delta n^{-1/d},
	\]
	so the tester rejects, as required.

	Now, to take a sample from $P$ (or $Q$), we simply take a sample
	$x$ from $p$ (or $q$) and return $f(x)$. Hence, the sample complexity of this
	tester is
	\[
		o\left(\left(\frac{\Delta}{(\frac{1}{2} \Delta n^{-1/d}}\right)^{2d/3}\right)
			= o(n^{2/3}).
	\]
	But this contradicts the lower bound for $\ell_1$-closeness testing from
	\cite{L1tester,valiant}, completing our proof.
\end{proof}

If one of the distributions is known, we can similarly
use the weaker $\Omega(n^{1/2})$ lower bound for uniformity testing to obtain the
following theorem.

\begin{theorem}
	There is no EMD-closeness tester that works on any $M \subset [0,\Delta]^d$, where
	one of the input distributions is explicitly known, that takes $o((\Delta/\eps)^{d/2})$
	samples.
\end{theorem}

%
%
%
%

\section{Clusterable distributions}

Since the general technique of our algorithms is to forcibly divide the input
distributions into several small ``clusters,'' it is natural to consider what
improvements are possible when the distributions are naturally clusterable.
We obtain the following substantial improvement from the exponential dependence
on the dimension $d$ that we had in the general case.

\begin{theorem}\label{thm:knowncenters}
	If the combined support of distributions $p$ and $q$ can be partitioned into
	$k$ clusters of diameter $\eps/2$, and we are given the $k$ centers, then
	there exists an EMD-closeness tester for $p$ and $q$ that requires
	only $\tilde{O}(k^{2/3}(d\Delta/\eps)^4)$ samples.
\end{theorem}

\begin{proof}
Let us denote the set of centers by $\mathcal{C} = \{C_1,\dots,C_k\}$. Consider the
distributions $P$ and $Q$ on $\mathcal{C}$ induced by $p$ and $q$, respectively,
by assigning each point to its nearest center. If $EMD(p,q) > \eps$,
by Lemma \ref{lem:buckets}, $\frac{\|P-Q\|_1}{2} (d\Delta) > \eps/2$. We can,
of course, obtain samples from $P$ and $Q$ by sampling from $p$ and $q$, respectively, and
returning the nearest center. Our problem thus reduces to $\ell_1$-testing for
$(\frac{\eps}{d\Delta})$-closeness over $k$ points, which requires
$\tO(k^{2/3}(d\Delta/\eps)^4)$ samples using the $\ell_1$-tester from
\cite{L1tester}.
\end{proof}

If we do not assume knowledge of the cluster centers, then we are still able to
obtain the following only slightly weaker result.

\begin{theorem}\label{thm:unknowncenters}
	If the combined support of $p$ and $q$ can be partitioned into $k$ clusters
	of diameter $\eps/4$, then even without knowledge of the centers there exists
	an EMD-closeness tester for $p$ and $q$ that requires only
	$\tO(kd\Delta/\eps + k^{2/3}(d\Delta/\eps)^4) \le \tO(k(d\Delta/\eps)^4)$ samples.
\end{theorem}

To prove this, we need the following result by Alon et al. that was implicit in
Algorithm 1 from \cite{cluster}.

\begin{lemma}(Algorithm 1 from \cite{cluster})
	There exists an algorithm which, given distribution $p$, returns $k' \le k$
	representative points if $p$ is $(k,b)$-clusterable, or rejects with
	probability $2/3$ if $p$ is $\gamma$-far from $(k,2b)$-clusterable, and
	which requires only $O(k\log k/\gamma)$ samples from $p$. Moreover, if the
	$k'$ points are returned, they are with probability $2/3$ the centers of a
	$(k,2b)$-clustering of all but a $\gamma$-weight of $p$.
\end{lemma}

\begin{proof}\emph{(of Theorem \ref{thm:unknowncenters})\ }
By the lemma, if our distributions are $(k,\eps/4)$-clusterable, using
$\tO(kd\Delta/\eps)$ samples we obtain a $(k',\eps/2)$-clustering of all but an
$\frac{\eps}{4d\Delta}$-fraction of the support of $p$ and $q$, with
centers $\mathcal{C}'$. Note that the unclustered probability mass
contributes at most $\eps/4$ to the EMD. The theorem then follows from an identical
argument as that of Theorem \ref{thm:knowncenters} (since we now know the centers).
\end{proof}

Note that in general, if we assume $(k,b)$-clusterability, this implies
$((2b/\eps)^dk,\eps/2)$-clusterability (by packing $\ell_1$ balls), where knowledge
of the super-cluster centers also implies knowledge of the sub-cluster centers.
Similarly, in the unknown centers case, $(k,b)$-clusterability implies
$((8b/\eps)^dk,\eps/4)$-clusterability. Unfortunately, in both cases we
reintroduce exponential dependence on $d$, so clusterability only really helps
when the cluster diameters are as assumed above.

\section{EMD over tree-metrics}

So far we have considered only underlying $\ell_1$-spaces (or, almost equivalently,
$\ell_p$). We will now see what can be done for EMD over tree-metrics, and prove
the following result.

\begin{theorem}
	If $p$ and $q$ are distributions over the nodes of a tree $T$, with
	edge weight function $w(\cdot)$, then there exists an $\eps$-additive-error
	estimator for $EMD(p,q)$ that requires only $\tilde{O}((Wn/\eps)^2)$
	samples, where $W = \max_e w(e)$, $n$ is the number of nodes of $T$, and $EMD(p,q)$ is defined with
	respect to the tree metric of $T$. Moreover, up to polylog factors,
	this is optimal.
\end{theorem}

\begin{proof}
First, let us consider an unweighted tree $T$ over $n$ points (i.e., where
every edge has unit weight), with distributions $p$ and $q$ on the
vertices. Observe that the minimum cost flow between $p$ and $q$ on $T$ is
simply the flow that sends through each edge $e$ just enough to balance $p$
and $q$ on each subtree on either side of $e$. In other words, if $T_e$ is
an arbitrary one of the two trees comprising $T-e$,
\[
        EMD(p,q) = \sum_e |p(T_e) - q(T_e)|.
\]
Then, with $\tO(n^2/\eps^2)$ samples we can, for every $T_e$, estimate
$p(T_e)$ and $q(T_e)$ to within $\pm \frac{\eps}{2(n-1)}$. This
gives us an $\eps$-additive estimator for $EMD(p,q)$.

Generalizing to the case of a weighted tree, where edge $e$ has weight
$w(e)$, we have
\[
        EMD(p,q) = \sum_e w(e)|p(T_e) - q(T_e)|.
\]
It then suffices to estimate each $p(T_e)$ and $q(T_e)$ term to within
$\pm \frac{\eps}{2w(e)(n-1)}$. Thus, $\tO((Wn/\eps)^2)$ samples
suffice, where $W = \max_e w(e)$.

Note that what we we get is
not only a closeness tester but also an additive-error estimator.
In fact, even if we only want a tester, this is the best we can do: in the case where $T$ is a line graph (with diameter $n-1$), the standard biased-coin lower bound implies we need $\Omega(n^2/\eps^2)$ samples.
\end{proof}

%

\bibliographystyle{plain}

\end{document}

%% file: preamble.tex

\hbadness=10000
\vbadness=10000

\setlength{\oddsidemargin}{.25in}
\setlength{\evensidemargin}{.25in}
\setlength{\textwidth}{6in}
\setlength{\topmargin}{-0.4in}
\setlength{\textheight}{8.5in}

\newcommand{\handout}[5]{
   \noindent
   \begin{center}
   \framebox{
      \vbox{
    \hbox to 5.78in { {\bf 6.896 Sublinear Time Algorithms}
     	 \hfill #2 }
       \vspace{4mm}
       \hbox to 5.78in { {\Large \hfill #5  \hfill} }
       \vspace{2mm}
       \hbox to 5.78in { {\it #3 \hfill #4} }
      }
   }
   \end{center}
   \vspace*{4mm}
}

\newcommand{\lecture}[4]{\handout{#1}{#2}{Lecturer:
#3}{Scribe: #4}{Lecture #1}}

\newtheorem{theorem}{Theorem}
\newtheorem{corollary}[theorem]{Corollary}
\newtheorem{lemma}[theorem]{Lemma}
\newtheorem{observation}[theorem]{Observation}
\newtheorem{proposition}[theorem]{Proposition}
\newtheorem{definition}[theorem]{Definition}
\newtheorem{claim}[theorem]{Claim}
\newtheorem{fact}[theorem]{Fact}
\newtheorem{assumption}[theorem]{Assumption}

\newcommand{\qed}{\rule{7pt}{7pt}}
\newcommand{\dis}{\mathop{\mbox{\rm d}}\nolimits}
\newcommand{\per}{\mathop{\mbox{\rm per}}\nolimits}
\newcommand{\area}{\mathop{\mbox{\rm area}}\nolimits}
\newcommand{\cw}{\mathop{\rm cw}\nolimits}
\newcommand{\ccw}{\mathop{\rm ccw}\nolimits}
\newcommand{\DIST}{\mathop{\mbox{\rm DIST}}\nolimits}
\newcommand{\OP}{\mathop{\mbox{\it OP}}\nolimits}
\newcommand{\OPprime}{\mathop{\mbox{\it OP}^{\,\prime}}\nolimits}
\newcommand{\ihat}{\hat{\imath}}
\newcommand{\jhat}{\hat{\jmath}}
\newcommand{\abs}[1]{\mathify{\left| #1 \right|}}

\newenvironment{proof}{\noindent{\bf Proof}\hspace*{1em}}{\qed\bigskip}
\newenvironment{proof-sketch}{\noindent{\bf Sketch of Proof}\hspace*{1em}}{\qed\bigskip}
\newenvironment{proof-idea}{\noindent{\bf Proof Idea}\hspace*{1em}}{\qed\bigskip}
\newenvironment{proof-of-lemma}[1]{\noindent{\bf Proof of Lemma #1}\hspace*{1em}}{\qed\bigskip}
\newenvironment{proof-attempt}{\noindent{\bf Proof Attempt}\hspace*{1em}}{\qed\bigskip}
\newenvironment{proofof}[1]{\noindent{\bf Proof}
of #1:\hspace*{1em}}{\qed\bigskip}
\newenvironment{remark}{\noindent{\bf Remark}\hspace*{1em}}{\bigskip}


\newcommand{\FOR}{{\bf for}}
\newcommand{\TO}{{\bf to}}
\newcommand{\DO}{{\bf do}}
\newcommand{\WHILE}{{\bf while}}
\newcommand{\AND}{{\bf and}}
\newcommand{\IF}{{\bf if}}
\newcommand{\THEN}{{\bf then}}
\newcommand{\ELSE}{{\bf else}}

\makeatletter
\def\fnum@figure{{\bf Figure \thefigure}}
\def\fnum@table{{\bf Table \thetable}}
\long\def\@mycaption#1[#2]#3{\addcontentsline{\csname
  ext@#1\endcsname}{#1}{\protect\numberline{\csname 
  the#1\endcsname}{\ignorespaces #2}}\par
  \begingroup
    \@parboxrestore
    \small
    \@makecaption{\csname fnum@#1\endcsname}{\ignorespaces #3}\par
  \endgroup}
\def\mycaption{\refstepcounter\@captype \@dblarg{\@mycaption\@captype}}
\makeatother

\newcommand{\figcaption}[1]{\mycaption[]{#1}}
\newcommand{\tabcaption}[1]{\mycaption[]{#1}}
\newcommand{\head}[1]{\chapter[Lecture \##1]{}}
\newcommand{\mathify}[1]{\ifmmode{#1}\else\mbox{$#1$}\fi}
\newcommand{\bigO}O
\newcommand{\set}[1]{\mathify{\left\{ #1 \right\}}}
\def\half{\frac{1}{2}}


\newcommand{\enc}{{\sf Enc}}
\newcommand{\dec}{{\sf Dec}}
\newcommand{\E}{{\rm Exp}}
\newcommand{\Var}{{\rm Var}}
\newcommand{\Z}{{\mathbb Z}}
\newcommand{\F}{{\mathbb F}}
\newcommand{\integers}{{\mathbb Z}^{\geq 0}}
\newcommand{\R}{{\mathbb R}}
\newcommand{\Q}{{\cal Q}}
\newcommand{\eqdef}{{\stackrel{\rm def}{=}}}
\newcommand{\from}{{\leftarrow}}
\newcommand{\vol}{{\rm Vol}}
\newcommand{\poly}{{\rm poly}}
\newcommand{\ip}[1]{{\langle #1 \rangle}}
\newcommand{\wt}{{\rm wt}}
\renewcommand{\vec}[1]{{\mathbf #1}}
\newcommand{\mspan}{{\rm span}}
\newcommand{\rs}{{\rm RS}}
\newcommand{\RM}{{\rm RM}}
\newcommand{\Had}{{\rm Had}}
\newcommand{\calc}{{\cal C}}
\renewcommand{\binom}[2]{{#1 \choose #2}}

\newcommand{\fig}[4]{
        \begin{figure}
        \setlength{\epsfysize}{#2}
        \vspace{3mm}
        \centerline{\epsfbox{#4}}
        \caption{#3} \label{#1}
        \end{figure}
        }

\newcommand{\ord}{{\rm ord}}

\providecommand{\norm}[1]{\lVert #1 \rVert}
\newcommand{\embed}{{\rm Embed}}
\newcommand{\qembed}{\mbox{$q$-Embed}}
\newcommand{\calh}{{\cal H}}
\newcommand{\lp}{{\rm LP}}